%% file: mainCDC2026.tex
\documentclass[letterpaper, 10pt, conference]{ieeeconf}  
\usepackage{amsmath,amssymb,amsfonts}

\usepackage{amsthm}
\usepackage[noadjust]{cite}
\usepackage{graphicx}
\usepackage[table]{xcolor}
\let\labelindent\relax
\usepackage{enumitem}
\usepackage{lipsum}
\usepackage{float}
\usepackage{array}
\usepackage{booktabs}
\usepackage{multirow}
\usepackage{siunitx}
\usepackage{url}
\usepackage{mathrsfs}
\usepackage{ragged2e}
\usepackage{tikz}
\usepackage{pgfplots}
\usepackage{algorithm, algpseudocode}
\usepackage{bm}

\IEEEoverridecommandlockouts
\title{\LARGE \bf
A Finite-Gain Stability Approach to NMPC Design: the Extended Version}

\author{Carlo Novara, Mattia Boggio, Lorenzo Calogero, Michele Pagone 
\thanks{
    The work of C. Novara and M. Pagone was supported by the Space It Up project, funded by the Italian Space Agency (ASI) and the Ministry of University and Research (MUR) under contract n. 2024-5-E.0, CUP n. I53D24000060005.}
\thanks{
    The authors are with the Department of Electronics and Telecommunications, Politecnico di Torino, 10129 Turin, Italy (e-mails: \{carlo.novara, mattia.boggio, lorenzo.calogero, michele.pagone\}@polito.it)}
}

\newtheorem{theorem}{Theorem}
\newtheorem{lemma}{Lemma}

\newtheorem{remark}{Remark}
\newtheorem{assumption}{Assumption}
\newtheorem{definition}{Definition}

\newcommand{\vs}{\vspace}

\newcommand{\blista}{\renewcommand{\labelenumi}{(\roman{enumi})} % list with roman numbers.
	\begin{enumerate}}
	\newcommand{\elista}{\end{enumerate} \renewcommand{\labelenumi}{\arabic{enumi}.}}

\allowdisplaybreaks

\begin{document}
\maketitle
\pagestyle{empty}
\thispagestyle{empty}

\begin{abstract}
This paper proposes a novel approach to design of Nonlinear Model Predictive Control (NMPC) schemes based on Finite-Gain Stability (FGS) concepts. 
The proposed formulation considers the case where the plant is affected by unknown but bounded disturbances, which renders difficult the classical Lyapunov-based analysis/design. 
Based on FGS conditions for a closed-loop system, we develop a systematic NMPC design methodology, allowing us to choose the relevant NMPC parameters that lead to closed-loop FGS and provide a satisfactory tracking performance, also for the case of time-varying reference signals. A simulated example is presented to demonstrate the effectiveness of our framework, concerned with lateral/longitudinal control of an automated vehicle.

\end{abstract}

\section{Introduction}
In the framework of optimal control, Nonlinear Model Predictive Control (NMPC) has emerged as one of the most effective and reliable control methodologies for nonlinear systems. Its success stems from its inherent ability to provide optimal control commands for multi-variable systems in the presence of inputs, outputs, and state constraints \cite{rawlings2017}.

In this context, establishing conditions for closed-loop stability is a nontrivial problem that poses exceptional challenges. This problem is even more relevant in the presence of disturbances, whose effects on the system state and tracking error are difficult to characterize. It is then natural to ask under which condition the \emph{nominal} NMPC\footnote{Henceforth, with \emph{nominal} NMPC we shall denote a \emph{non-robust} controller which, in principle, is still able to guarantee the stability of the closed-loop in presence of external disturbance.} is able to ensure that perturbed trajectories remain bounded or converge close to the desired reference. Such problem is, in general, difficult to characterize.

Nowadays, nominal NMPC  stabilizing design theory is typically based on the  Lyapunov approach, see, e.g. \cite{Mayne2013}, and the reference therein. Nevertheless, its extension to the case where the plant is affected by unknown disturbances and time varying references is nontrivial. Early works (see, e.g., \cite{teel2004,kellet2004}) provide Lyapunov-based stability certificates without, however, offering constructive guidelines for determining disturbance bounds and for explicitly quantifying their effect on the tracking performance, making the stabilizing design rather conservative and hard to apply in practice. 
%Furthermore, the search for a control Lyapunov function to serve as terminal cost and/or the admissible forward-invariant terminal sets may be harsh and mathematically demanding \cite{Gilbert:Tan:1991:Admissible,Kerrigan:00}.

Alternative approaches, which can be more convenient than the standard Lyapunov-based procedure in certain settings, formulate the closed-loop stability of nonlinear systems by leveraging the notions of Input-to-State Stability (ISS) \cite{Limon2009} and Finite-Gain Stability (FGS) \cite{ZACCARIAN2003,book-khalil-2002}.
%Alternative approaches for formulating the closed-loop stability of nonlinear systems employs the (similar) concepts of Input-to-State Stability (ISS) \cite{limon2002, limon2006, Limon2009} and Finite-gain Stability (FGS) \cite{ZACCARIAN2003,book-khalil-2002}. 
Both these concepts characterize how the system state/output converges or remains bounded, depending on the initial condition, the magnitude of the reference, and the bounds on the disturbance. 

Building up on the recent results on ISS and FGS, in this paper, we propose a novel FGS-based design approach to NMPC control systems. We establish sufficient conditions under which, given a set of feasible initial conditions and possibly time-varying references, under the application of the NMPC control law, the closed-loop trajectories remain bounded and converge towards a  neighborhood centered on the reference.

The proposed framework is not meant to serve as a robust counterpart of the nominal NMPC (as in \cite{limon2006, Limon2009}), rather, it defines conditions under which a NMPC controller is inherently capable to attenuate the external disturbance for letting the system's output to converge and remain bounded in a close neighborhood of the reference. %(without, however, zero-offset asymptotically convergence).
The stability of the closed-loop is not guaranteed through the standard terminal ingredients of the Lyapunov framework \cite{Mayne2013}, but follows from the FGS bound estimates. 
However, unlike the classical ISS-based and FGS-based methods, see, e.g., \cite{limon2002, Alessandretti2017}, thanks to a quasi-LPV (Linear Parameter Varying) formulation, we explicitly deal with the case of time-varying reference, but in a different fashion with respect to the artificial reference tracking MPC approach, see, e.g., \cite{KOHLER2024100929,krupa2024} and the references therein. Indeed, in this latter approach, the terminal constraints/penalties may be hard to design for arbitrary time-varying references. Instead, we formulate FGS concepts using a tracking-oriented error bound that directly accounts for reference variations.

Hence, we focus on the choice of the NMPC configuration parameters, namely the prediction horizon and the weighting matrices appearing in the cost function. Given a region of initial conditions, a class of reference signals, the goal is to select an NMPC configuration that provides both accurate tracking and suitable stability properties when the plant is subject to unknown but bounded disturbances. 

To this end, we propose an FGS-based three-step offline NMPC design procedure: i) a finite simulation horizon is selected so as to capture the transient interval over which the contraction of the closed-loop error dynamics becomes visible; ii) a family of candidate NMPC configurations is generated by varying the main tuning parameters around a nominal design; iii) each candidate is evaluated over a scenario set composed of initial conditions and reference signals. In this way, the FGS analysis is used as a practical guide for controller synthesis.

The appeal of the proposed methodology stems from the fact that each candidate NMPC configuration is assessed only over a finite horizon, using quantities that are meaningful both from a stability and from a performance viewpoint. 
%This makes the procedure efficient from a computational point of view, and attractive when one seeks a practical and interpretable way to tune NMPC controllers over a prescribed operating domain. 
This makes the procedure computationally efficient and particularly suitable for the practical tuning of NMPC controllers over a prescribed operating domain, while preserving a high level of interpretability.
The novel contribution of the paper is twofold: i) we derive finite-gain bounds for the tracking error of the closed-loop NMPC system; ii) we exploit this bound to formulate an offline stabilizing design methodology for selecting the NMPC parameters. 

The proposed approach is illustrated on an automated-vehicle control problem involving lateral and longitudinal motion, where the method is used to choose the prediction horizon and cost weights before validating the resulting controller on lane-keeping and obstacle-avoidance maneuvers.

\subsection{Notation}

In the following, $x = (x_1, \ldots, x_n) = (x_i)_{i=1}^n \in \mathbb{R}^n$ denotes a column vector $x$, with components $x_i$; $x = [x_1, \ldots, x_N] = [x_i]_{i=1}^n = (x_1, \ldots, x_n)^\top \in \mathbb{R}^{1 \times n}$ denotes a row vector; $\|x\|$ is the $2$-norm of vector $x$; $\|x\|_M = \|M^{1/2}x\| = \sqrt{x^\top M x}$ is the weighted norm of vector $x$, with weighting matrix $M$; $\|A\|$ is the induced $2$-norm of matrix $A$.

\input{nmpc_ecc_1.tex}

\input{fgs_nominal_ecc_2.tex}

\input{nmpc_design_2.tex}

\section{Example: Lateral/longitudinal control for an automated vehicle} \label{sec:results}
We present the application of the proposed NMPC design methodology to the problem of controlling the lateral and longitudinal dynamics of an automated vehicle. The goal is to show how the systematic procedure described in Section~\ref{sec:nmpc_des} leads to the selection of an NMPC configuration that gives closed-loop finite-gain stability and satisfactory tracking performance. The resulting controller is then validated on two standard maneuvers: lane keeping and obstacle avoidance. 
\subsection{Vehicle and prediction model} \label{sec:vehicle_model}
% The same model is used for both simulating the vehicle dynamics and as the prediction model within the NMPC. In particular, 
A standard Dynamic Single-Track (DST) model is considered, providing a representation of both lateral and longitudinal vehicle dynamics. Although simpler than other models, it captures the key aspects required for designing and conducting preliminary tests on vehicle control systems. The state equations of the DST model are as follows:
\vspace*{-0.4em}
\begin{equation}
\begin{aligned}
&\dot{p}_X =v_{x}\cos\psi-v_{y}\sin\psi\\
&\dot{p}_Y =v_{x}\sin\psi+v_{y}\cos\psi\\
&\dot{\psi}=\omega\\
&\dot{v}_{x} =v_{y}\dot{\psi}+a_x\\
&\dot{v}_{y} =-v_{x}\dot{\psi}+\frac{2}{m}\left(F_{yf}+F_{yr}\right)\\
&\dot{\omega} =\frac{2}{I_z}\left(l_{f}F_{yf}-l_{r}F_{yr}\right)
\end{aligned}
\label{eq:dstp_se}
\vspace*{-0.4em}
\end{equation}
where $p_X$ and $p_Y$ are the coordinates of the vehicle in an inertial frame, $\psi$ is the yaw angle, $\omega$ denotes the yaw rate, and $v_x$, $v_y$ are the longitudinal and lateral speeds, respectively. The main physical parameters of the model and their values are as follows: mass $m=1575\,\text{kg}$; moment of inertia $I_z=4000\,\text{kgm}^2$; distances of Center of Gravity (CoG) to front wheels $l_f=1.2\,\text{m}$ and rear wheels $l_r=1.6\,\text{m}$. Moreover, $F_{yf}$ and $F_{yr}$ are the lateral forces, given by
$F_{yf}=-c_f\beta_f, \, F_{yr}=-c_r\beta_r$, 
where $c_f=2.7\cdot10^{4}\,\text{N/rad}$ and $c_r=2\cdot10^{4}\,\text{N/rad}$ are the front/rear cornering stiffnesses. The front/rear tire slip angles $\beta_f$ and $\beta_r$ are defined as
\vspace*{-0.4em}
\begin{equation}
\beta_{f}=\mathrm{atan}\left(\frac{v_{y}+l_{f}\dot{\psi}}{v_{x}}\right)-\delta_{f}, \,\,
\beta_{r}=\mathrm{atan}\left(\frac{v_{y}-l_{r}\dot{\psi}}{v_{x}}\right).
\end{equation}
The longitudinal acceleration $a_x$ and the steering angle $\delta_f$ are the control variables with the following saturation constraints:
$a_x \in [-5,\, 3]~\mathrm{m/s^2}, \delta_f \in [-0.78,\, 0.78]~\mathrm{rad}.$
The DST equations have been discretized by means of the Forward Euler method, using the sampling time $Ts = 0.1$ s. A discrete-time model of the form \eqref{eq:dstp_se} has been obtained and used for both simulating the vehicle dynamics and as the prediction model within the NMPC. The controlled outputs considered in the NMPC design are $(p_X,\,p_Y,\,\psi,\,v_x)$. To account for model–plant mismatch and measurement errors, a Gaussian noise term was added to the plant dynamics during simulation. Specifically, the noise had a mean vector  $\mu = [0.5 ,\,0.5,\,0.001,\,0.05,\,0.05,\,0.001]$ and covariance matrix $\Sigma = \mathrm{diag}(10^{-4},\,10^{-4},\,10^{-5},\,10^{-4},\,10^{-5},\,10^{-5})$.
% The longitudinal acceleration $a_x$ and the steering angle $\delta_f$ are the control variables. The DST equations have been discretized by means of the Forward Euler method, using the sampling time $Ts = 0.1$ s. A discrete-time model of the form \eqref{eq:dstp_se} has been obtained and used for both simulating the vehicle dynamics and as the prediction model within the NMPC.
%as the NMPC internal model.
\subsection{NMPC design}
The NMPC configuration is obtained by following the procedure detailed in Section~\ref{sec:nmpc_des}. 
%All closed-loop simulations and NMPC controllers use a common, fixed sampling time \(T_s = 0.1\ \mathrm{s}\). 
The scenario sets \(\mathcal{S}_{\mathrm{in}}\) and \(\mathcal{S}\) are generated using the Latin Hypercube Sampling (LHS) \cite{McKay1979}. LHS is a statistical
method that produces quasi-random samples of parameter
values from a multidimensional distribution, providing better
coverage of the domain compared to random sampling,
where each new sample is generated independently of previously chosen points. The initial state set \(X_0\) was defined as:
\[
X_0 = \{x_0 = (p_{X,0},p_{Y,0},\psi_0,v_{x,0},v_{y,0},\omega_0) \mid  
\]
\[
p_{X,0} \in [0,100]~\mathrm{m},\, p_{Y,0}, \in [0,100]~\mathrm{m},\, 
\psi_0 \in [-0.5,\,0.5]~\mathrm{rad},  
\]
\[
 v_{x,0} \in [5,\,15]~\mathrm{m/s},\;v_{y,0} \in [0.1,\,1]~\mathrm{m/s},\, 
\]
\[\omega_0 \in [-0.1,\,0.1]~\mathrm{rad/s}\}.\]
The reference set \(\mathcal{R}\) was generated using the same intervals for the corresponding reference variables (position, yaw angle and longitudinal velocity), ensuring that initial conditions and targets are sampled from a consistent domain.

\subsubsection{Selection of the simulation horizon $\tau^*$}
The first step consists in determining the value of $\tau^*$. The initial scenario set $\mathcal{S}_{\mathrm{in}}$ was constructed considering a total of $M=10$ scenarios. A nominal NMPC configuration 
$\phi_{\mathrm{nom}}$ with  $T_{\mathrm{nom}}=30$, $Q_{\mathrm{nom}}=\mathrm{diag}(1,1,0.1,1)$, $R_{\mathrm{nom}}=\mathrm{diag}(0.1,1)$ and $P_{\mathrm{nom}}=Q_{\mathrm{nom}}$
was selected. For each scenario $s^i \in \mathcal{S}_{\mathrm{in}}$, closed-loop simulations were performed until the vehicle reached the reference trajectory, with a total simulation duration of $150$ time steps. At each time instant $t = 1, \dots, T_{\mathrm{fin}}$, the finite-gain index $\gamma_i(t)$ was computed according to~\eqref{eq:stab_cond}. The smallest integer $\tau_i$ satisfying Eq. \eqref{eq:gamma_t} was then stored for each scenario. The resulting values of $\tau_i$ ranged between $8$ and $12$ sampling steps. To ensure robustness to untested conditions, the simulation horizon was set to $\tau^* = 15$. 
This value of $\tau^*$ was subsequently used in the next steps of the NMPC design procedure. 

\subsubsection{Closed-loop Function Sampling campaign}
An evaluation scenario set $\mathcal{S}$ of $N=20$ scenarios was generated. Then, a function sampling campaign was conducted to evaluate a set of $K=160$ candidate NMPC configurations. The prediction horizon was varied over $T\in\{5,\,10,\,15,\,20,\,25,\,30,\,35,\,40,\,45,\,50\},
$ and the weighting matrices were obtained by scaling the nominal NMPC configuration $\phi_{\mathrm{nom}}$:
\[
Q=\alpha_Q Q_{\mathrm{nom}},\,\, R=\alpha_R R_{\mathrm{nom}},\,\, \alpha_Q,\alpha_R\in\{0.1,\,1,\,10,\,100\}.
\]
For simplicity, the terminal cost matrix $P$ was set equal to $Q$, i.e., $\alpha_P=\alpha_Q$. For each configuration $\phi_j$, closed-loop simulations were performed over the finite simulation horizon $\tau^*$ and for all scenarios in $S$, resulting in more than $3000$ simulations. The tracking error $E(\phi_j,s^i)$ and the finite-gain index $L(\phi_j,s^i)$ were stored according to~\eqref{eq:E_c} and~\eqref{eq:L_c}. The complete design procedure took about $40$ minutes, using a laptop with processor Intel(R) Core(TM) i7-10750H CPU @ 2.60GHz. %Compared to an our previous design procedure, based on long-horizon simulations, the proposed method reduced the total computational time of about one order of magnitude.

\subsubsection{Selection of the optimal NMPC configuration}
After computing the normalized indices $E_n(\phi)$ and $L_n(\phi)$, the scalar cost function
\[
J_{sr}(\phi) = \alpha_J E_n(\phi) + (1-\alpha_J)L_n(\phi)
\]
was minimized with $\alpha_J=0.5$, providing a good balance between tracking accuracy and stability. The resulting optimal NMPC configuration $\phi^\star$ is $T^\star = 20, \,
\alpha_Q^\star = \alpha_P^\star=100, \,
\alpha_R^\star = 1.
$

\subsubsection{Visualization and sensitivity analysis}
To illustrate the influence of design parameters, level curve plots of the metrics $E(\phi)$ and $L(\phi)$ were generated over the grid $(T,\kappa)$. An illustrative example of the resulting level curve of $L(\phi)$ for $R=1$ and varying $T$ and $Q$ is shown in Figure~1. This figure provides a visual map of how the stability margin changes as the prediction horizon and the relative state-to-input weighting are varied. The white contours denote iso-values of $L(\phi)$, while the color map highlights the corresponding numerical levels. Regions associated with smaller values of $L(\phi)$ indicate NMPC configurations with a stronger contraction property and, therefore, a more favorable finite-gain stability margin. The figure shows that the best configurations are concentrated in limited portions of the grid, suggesting that the choice of $T$ and $\kappa$ has a nontrivial effect on closed-loop robustness.
%\vspace*{-0.6em}
\begin{figure}[htb]
\begin{center}
\includegraphics{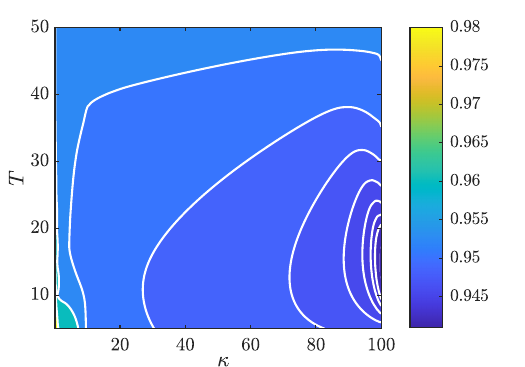}  
\caption{Level curves of $L(\phi)$.} 
\label{fig:obs_avoid}
\end{center}
\end{figure}
\vspace*{-1em}

\subsection{Lane Keeping}
The selected NMPC configuration $\phi^\star$ was first validated on a lane-keeping scenario and compared with an alternative configuration $\phi_a$ with $T=40$, $\alpha_Q=\alpha_P=20$ and $\alpha_R=1$. 
Both controllers were tested on a curvilinear urban road, modeled by a sinusoidal path of the form
$r = 8 \sin(0.02\,t),$
corresponding to moderate lateral curvature variations. The longitudinal velocity reference was set to a constant value of $v_x^{\mathrm{ref}} = 40~\mathrm{km/h}$. The control inputs and the saturation limits are defined in Section \ref{sec:vehicle_model}.
Table~\ref{tab:lane_results} summarizes the key tracking error metrics. The results indicate that \(\phi^\star\) achieves marginally better lateral tracking and smaller orientation deviations than \(\phi_a\). The improved precision of \(\phi^\star\) is consistent with its high state weighting (\(\alpha_Q^\star=100\)).

\begin{table}[htb]
\centering
\caption{Comparison between \(\phi^\star\) and \(\phi_a\).}
\label{tab:lane_results}
\begin{tabular}{lcc}
\toprule
 & \(\phi^\star\) & \(\phi_a\) \\
\midrule
RMS lateral error [m] & \(0.025\) & \(0.029\) \\
Maximum lateral error [m] & \(0.07\) & \(0.1\) \\
RMS orientation error [rad] & \(0.0005\) & \(0.00055\) \\
Maximum orientation error [rad] & \(0.0007\) & \(0.0012\) \\
\bottomrule
\end{tabular}
\end{table}

\subsection{Obstacle Avoidance}

The second validation scenario involves an obstacle-avoidance maneuver designed to test the capability of the NMPC to ensure collision-free operation while preserving stability and maintaining acceptable tracking accuracy after the avoidance phase. 
%Obstacles are represented by nonlinear state constraints of the form:
% \begin{equation}
% \frac{(X - c_{X})^2}{(l_X/2)^2} +
% \frac{(Y - c_{Y})^2}{(l_{Y}/2)^2}
% \geq 1.
% \label{eq:obs_constraints}
% \end{equation}
The obstacle is modeled as an elliptical safety region with center coordinates $(c_{X}, c_{Y}) = (5, 5)$ and sizes $(l_{X}, l_{Y}) = (1, 1)$. As shown in Figure~\ref{fig:obs_avoid},  both NMPC configurations successfully executed the avoidance maneuver without violating the safety constraints. However, notable qualitative differences were observed. The optimal configuration $\phi^\star$ maintained the vehicle trajectory closer to the nominal path connecting the initial and final points, thus minimizing lateral deviation during the maneuver. This behaviour reflects the higher weighting of the controller state and its shorter prediction horizon, which together enable faster corrective actions and improved reactivity.
%\vspace*{-1em}
\begin{figure} [htb]
\begin{center}
\includegraphics{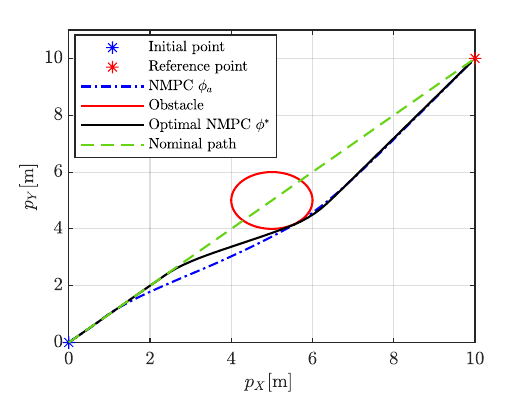}   
\caption{Obstacle avoidance scenario.} 
\label{fig:obs_avoid}
\end{center}
\end{figure}
%\vspace*{-1em}
\section{Conclusions} \label{sec:conclusion}
This paper has presented a finite-gain-based approach to the stabilizing design of NMPC control systems. The main result is a bound on the closed-loop tracking error under a finite-step contraction condition on the trajectory-dependent error dynamics, which makes explicit the role of the initial error, the disturbance, and the reference variation. An offline scenario-based design procedure has been proposed to select the prediction horizon and weighting matrices by combining a tracking-performance index with a finite-gain stability index computed over a finite simulation horizon. The method has been validated on an automated-vehicle example, where it identified a controller configuration with satisfactory tracking and stability properties over the vehicle maneuvers. 

\section*{Appendix}

\subsection{Proof of Lemma~\ref{lem:ltv_err}}

Consider the closed-loop system~\eqref{eq:cls1}. Its tracking error at time $t+1$ is given by
\begin{align}
    e_{t+1} = \; & r_{t+1}-x_{t+1} = r_{t+1}-f_{c}(x_{t},\bm{r}_{t})-d_{t} \nonumber \\
    = \; & r_{t+1}-f_{c}(r_{t},\bm{r}_{t})-d_{t} + f_{c}(r_{t},\bm{r}_{t})-f_{c}(x_{t},\bm{r}_{t}).
\end{align}

From the definition of the matrix $A_{t}$ in Eq.~\eqref{eq:ltv1}, we have
\begin{align}
    A_{t}(r_{t}-x_{t}) = \; & \frac{(f_{c}(r_{t},\bm{r}_{t})-f_{c}(x_{t},\bm{r}_{t}))(r_{t}-x_{t})^{\top}(r_{t}-x_{t})}{\left\| r_{t}-x_{t} \right\|^{2}} \nonumber \\
    = \; & f_{c}(r_{t},\bm{r}_{t})-f_{c}(x_{t},\bm{r}_{t}),
\end{align}
which implies
\begin{align}
    f_{c}(r_{t},\bm{r}_{t})-f_{c}(x_{t},\bm{r}_{t})=A_{t}(r_{t}-x_{t})=A_{t}e_{t}.
\end{align}

The claim follows. \hfill $\square$

\subsection{Proof of Theorem \ref{thm:fgs_nom}.}

Define $v_{t}=\xi_{t}-d_{t}=r_{t+1}-f_{c}(r_{t},\bm{r}_{t})-d_{t}$
and iterate the tracking error dynamics~\eqref{eq:treev}, i.e.,
\begin{align}
    e_{t+2} = \; & A_{t+1}e_{t+1}+v_{t+1} \nonumber = \; A_{t+1}A_{t}e_{t}+A_{t+1}v_{t}+v_{t+1} \\[1ex]
    e_{t+3} = \; & A_{t+2}e_{t+2}+v_{t+2} \nonumber \\
    = \; & A_{t+2}A_{t+1}A_{t}e_{t}+A_{t+2}A_{t+1}v_{t} + A_{t+2}v_{t+1}+v_{t+2} \\
    \vdots \; & \nonumber
\end{align}

Considering the transition matrix $F_{t}^{\tau}$ in Eq.~\eqref{eq:ltv1}, we have
\begin{align}
    e_{t+\tau} = \; & F_{t}^{\tau}e_{t}+\sum_{i=t}^{t+\tau-1}F_{i+1}^{t+\tau}v_{i}=F_{t}^{\tau}e_{t}+\mu_{t}
\end{align}
where $\mu_{t}=\sum_{i=t}^{t+\tau-1}F_{i+1}^{t+\tau}v_{i}$.
We define the subsequence $\check{e}_{l}= e_{l\tau}$, $l \in \mathbb{Z}_{\geq 0}$,
whose dynamics is described by
\begin{align} \label{eq:subseq}
    \check{e}_{l+1} = \check{F}_{l}^{\tau}\check{e}_{l}+\mu_{l}
\end{align}
where $\check{F}_{l}^{\tau} = F_{l\tau}^{\tau}$.
From Assumption~\ref{assu:tm_bound}, given that $\gamma=\max_{t \geq 0, \, x_0 \in \mathcal{X}_{0}, \, \bm{r} \in \mathcal{R}^\infty} \left\| F_{t}^{\tau}\right\| $,
the transition matrix
\begin{align}
    \Psi_{l}^{m} = \begin{cases}
        \check{F}_{l+m-1}^{\tau}\check{F}_{l+m-2}^{\tau}\ldots\check{F}_{l}^{\tau}, & m>0, \\
        I, & m=0
    \end{cases}
\end{align}
is bounded as $\left\| \Psi_{l}^{m}\right\| \leq\gamma^{m}$
for all $l, m \geq 0$. Now, consider that the solution of Eq.~\eqref{eq:subseq} is given by
\begin{align}
    \check{e}_{l} = \; & \Psi_{0}^{l}\check{e}_{0}+\sum_{j=0}^{l-1}\Psi_{j+1}^{l-j-1}\mu_{j}.
\end{align}
It follows that the subsequence $\check{e}_{l}$ is bounded as
\begin{align}
    \left\| \check{e}_{l}\right\| \leq \; & \left\| \Psi_{0}^{l}\right\| \left\| \check{e}_{0}\right\| +\sum_{j=0}^{l-1}\left\| \Psi_{j+1}^{l-j-1}\right\| \left\| \mu_{j}\right\| \nonumber \\
    \leq \; & \gamma^{l}\left\| \check{e}_{0}\right\| +\gamma^{l-1}\sum_{j=0}^{l-1}\gamma^{-j}\left\| \mu_{j}\right\|.
\end{align}
Moving back from the subsequence to the original sequence, we have that
\begin{align}
    \left\| e_{l\tau}\right\| \leq \; & \gamma^{l}\left\| e_{0}\right\| +\gamma^{l-1}\sum_{j=0}^{l-1}\gamma^{-j}\left\| \mu_{j}\right\| \nonumber \\
    \leq \; & \gamma^{l}\left\| e_{0}\right\| +\gamma^{l-1}\sum_{j=0}^{l-1}\gamma^{-j}\sum_{i=j}^{j+\tau-1}\left\| v_{i}\right\| \nonumber \\
    = \; & \gamma^{l}\left\| e_{0}\right\| +\gamma^{l-1}\sum_{i=0}^{l\tau-1}\gamma^{-\lfloor i/\tau \rfloor}\left\| v_{i}\right\|.
\end{align}
where $\lfloor\cdot\rfloor$ is the floor operator. This bound holds for $t=l\tau$, where $l \geq 0$.
For a generic $t \geq 0$, we observe that any sample $e_{t}$ can be obtained by a finite number of iterations of Eq.~\eqref{eq:treev}, starting from $e_{l\tau}$, with $l=\lfloor t/\tau \rfloor$.

We also observe that $\left\| A_{t}\right\| $ is bounded. Indeed, if, for some $t$, $\left\| A_{t}\right\| =\infty$, a trajectory would exist for which the condition $\left\| F_{t}^{\tau}\right\| <1$ would not hold, thus invalidating Assumption~\ref{assu:tm_bound}.

Since the number of iterations of Eq.~\eqref{eq:treev} is finite and $\left\| A_{t}\right\| $ is bounded for all $t$, it follows that a finite positive constant $L$ exists such that
\begin{align}
    \left\| e_{t} \right\| \leq \; & L\gamma^{l}\left\| e_{0}\right\| +L\gamma^{l-1}\sum_{i=0}^{t-1}\gamma^{-\lfloor i/\tau \rfloor}\left\| v_{i}\right\|.
\end{align}

Consider now that $l=\lfloor t/\tau \rfloor\geq t/\tau-1$, which
implies that $\gamma^{l}=\gamma^{\lfloor t/\tau \rfloor}\leq\gamma^{t/\tau-1}$.
This yields
\begin{align}
    \left\| e_{t}\right\| \leq \; & L\gamma^{t/\tau-1}\left\| e_{0}\right\| +L\gamma^{t/\tau-2}\sum_{i=0}^{t-1}\gamma^{-i/\tau}\left\| v_{i}\right\| \nonumber \\
    = \; & L\gamma^{t/\tau-1}\left\| e_{0}\right\| +L\gamma^{-2}\sum_{i=0}^{t-1}\gamma^{(t-i)/\tau}\left\| v_{i}\right\|.
\end{align}

Recalling that $v_{i}=\xi_{i}-d_{i}$ and defining $\rho=\gamma^{1/\tau}<1$,
we conclude that finite non-negative constants $\Lambda_{0}$, $\Lambda_{\xi}$
and $\Lambda_{d}$ exist, such that
\vspace*{-0.4em}
\begin{align}
    \left\| e_{t}\right\| \leq \; & \Lambda_{0}\rho^{t}\left\| e_{0}\right\| +\Lambda_{\xi}\sum_{i=0}^{t-1}\rho^{t-i}\left\| \xi_{i}\right\| +\Lambda_{d}\sum_{i=0}^{t-1}\rho^{t-i}\left\| d_{i}\right\|.
\end{align}
\vspace*{-0.4em}

This proves claim (ii). Claim (i) directly follows from this inequality, considering that the reference sequence is bounded. \hfill $\square$

%\section*{Acknowledgment}

\bibliographystyle{IEEEtran}
\bibliography{references.bib,lettaltr.bib}

\end{document}

%% file: nmpc_ecc_1.tex
\section{Discrete-time NMPC }

\label{sec:nmpc}

%\subsection{Plant to control}

Consider a discrete-time nonlinear system described by the following
state equation:
\begin{equation}
x_{t+1}=f(x_{t},u_{t})+d_{t}\label{eq:plant}
\end{equation}
where $t\in\mathbb{N}_{0}$ is the time index, $x_{t}\in\mathbb{R}^{n_{x}}$
is the system state, $u_{t}\in U\subset\mathbb{R}^{n_{u}}$ is the
command input, $d_{t}\in D\subset\mathbb{R}^{n_{d}}$ is an exogenous
disturbance, $U$ and $D$ are compact sets. The state is assumed
to be measured in real-time. If not measured, it can be estimated
by means of an observer or an input-output model can be used instead
of the state equations. 

%\textcolor{red}{Dobbiamo dire qualcosa su $X$? Teoricamente $X = \mathbb{R}^{n_x}\setminus B$ dove $B$ è l'ostacolo. Quindi in realtà $X$ non è compatto. Però poi introduciamo $X_0$ come compatto e connesso, ma solo per le condizioni iniziali. Forse vale la pena scrivre che esiste $X \subset \mathbb{R}^{n_x}\setminus B$ con $X_0 \subseteq X$. }

The control task is to make the system state $x_{t}$ track a desired
reference signal $r_{t}\in\mathbb{R}^{n_{x}}$. The state and input
variables may be subject to constraints and it may be of interest
to obtain a suitable trade-off between performance and command effort.
To this aim, a NMPC approach
is used. 

%\subsection{NMPC formulation}

The NMPC formulation developed in this paper is as follows. At each
time step $t$, the system state $x_{t}$ is measured. An optimal
command is computed based on $x_{t}$, and then applied to the system.
%The optimal command is obtained by means of two key operations: prediction and optimization. 
%\emph{Prediction.} At each time step $t$, the system state is predicted
%$T$ steps in the future, where $T$ is called the \emph{prediction
%horizon}. The $k$-step prediction is obtained by iterating $k$ times
%equation (\ref{eq:plant}). For any $k\in[0,T-1]$, the predicted
%state $\hat{x}_{k}$ is a function of the applied input sequence $(\mathfrak{u}_{0},\ldots,\mathfrak{u}_{k-1})$,
%where the time step ``$0$'' of the prediction interval corresponds
%to the time step $t$ of the plant evolution. 
Let $T\geq 1$ be the prediction horizon, the NMPC objective function is defined as 
\begin{equation}
J_{o}\left(x_t,\boldsymbol{\mathfrak{u}}\right)\doteq\sum_{k=0}^{T-1}\mathfrak{u}_{k}^{\top}R\mathfrak{u}_{k}+\sum_{k=1}^{T-1}\tilde{x}_{k}^{\top}Q\tilde{x}_{k}+\tilde{x}_{T}^{\top}P\tilde{x}_{T}\label{eq:obj}
\end{equation}
where $\tilde{x}_{k}\doteq\mathfrak{r}_{k}-\hat{x}_{k}$ is the predicted
tracking error, $\mathfrak{r}_{k}\doteq r_{t+k}\in\mathbb{R}^{n_{x}}$,
$k=1,\ldots,T$, is the reference to track, $\boldsymbol{\mathfrak{u}}\doteq(\mathfrak{u}_{0},\ldots,\mathfrak{u}_{T-1})\in\mathbb{R}^{Tn_{u}\times1}$
is the applied input sequence, and $R,Q,P$ are diagonal positive-definite weight matrices. 

The input constraints are described by $\mathfrak{u}_{k}\in U,\,k=1,\ldots,T$.
The state constraints are treated as \emph{soft constraints} using penalty functions. In particular,
for each constraint, a penalty function $\sigma_{j}\in\mathcal{C}^{1}$
is introduced, such that $\sigma_{j}(\hat{x}_{k})=0$ when the constraint
is satisfied, $\sigma_{j}(\hat{x}_{k})\gg J_{o}\left(\boldsymbol{\mathfrak{u}}\right)$
when it is not (see, e.g., \cite{pagone2024} and \cite{malisani2016}). The penalty functions are added to $J_{o}$, giving
the ``augmented'' objective function 
\begin{equation}
\begin{array}{c}
J\left(x_t,\mathfrak{u}\right)\doteq J_{o}\left(x_t,\mathfrak{u}\right)+\sum_{j,k}\sigma_{j}(\hat{x}_{k})\end{array}.\label{eq:Jtot}
\end{equation}
The optimal input sequence is computed at each time $t\in\mathbb{N}_{0}$,
solving the following optimization problem:
\begin{eqnarray}
 &  & \boldsymbol{\mathfrak{u}}^{*}=\arg\underset{\boldsymbol{\mathfrak{u}}}{\min}\,J\left(x_{t},\mathfrak{u}\right)\label{eq:opt}\\
 &  & \begin{array}{l}
\textrm{subject to:}\vs{1.5mm}\\
\qquad\hat{x}_{0}=x_{t}\vs{1.5mm}\\
\qquad\hat{x}_{k+1}={f}(\hat{x}_{k},\mathfrak{u}_{k}),\;k=0, \ldots,T-1\vs{1.5mm}\\
\qquad\mathfrak{u}_{k}\in U,\;k=0, \ldots,T-1.
\end{array}\label{eq:cons}
\end{eqnarray}

The plant is controlled by NMPC under the one-step receding horizon
policy: at each time instant $t \geq 0$, only the first optimal input
sample is applied to the plant. The remainder of the input sequence
is discarded.

%\subsection{Input sequence dimension reduction}

The optimization problem (\ref{eq:opt}) is in general non-convex.
Moreover, the command sequence $\boldsymbol{\mathfrak{u}}\doteq(\mathfrak{u}_{0},\ldots,\mathfrak{u}_{T-1})$
may contain a relatively large number of decision variables. An efficient
method to reduce the number of variables is the Move
Blocking technique, see, e.g.,~\cite{Tracking:MPC:Limon:Automatica:08, SHEKHAR2012587}: the command sequence $\boldsymbol{\mathfrak{u}}$
is parametrized as
\begin{equation}
\boldsymbol{\mathfrak{u}}=\Gamma c\label{eq:dimred}
\end{equation}
where $\Gamma\in\mathbb{R}^{Tn_{u}\times n_{c}}$ is a full-rank matrix
with $n_{c}<Tn_{u}$, and $c\in\mathbb{R}^{n_{c}}$ is a new command
input sequence with reduced dimension. A sample $\mathfrak{u}_{k}$
of $\boldsymbol{\mathfrak{u}}$ can be obtained using a selection
matrix $S_{k}^{u}$ such that $\mathfrak{u}_{k}=S_{k}^{u}\boldsymbol{\mathfrak{u}}=S_{k}^{u}\Gamma c$. Note that $\Gamma$ is a design parameter, selected a-priori.

% \begin{remark}
%     Note that, $\Gamma$ is treated as a fixed design choice, selected a-priori in order to ensure that the NMPC optimization problem remains feasible.
% \end{remark}

Under the above input parametrization, the optimization problem
becomes:
\begin{eqnarray}
 &  & c^{*}=\arg\underset{c}{\min}\,J\left(x_t,\Gamma c\right)\label{eq:opt-1}\\
 &  & \begin{array}{l}
\textrm{subject to:}\vs{1.5mm}\\ \nonumber
\qquad\hat{x}_{0}=x_{t}\vs{1.5mm}\\
\qquad\hat{x}_{k+1}={f}(\hat{x}_{k},S_{k}^{u}\Gamma c),\;k=0,\ldots,T-1\vs{1.5mm}\\
\qquad S_{k}^{u}\Gamma c\in U,\;k=0,\ldots,T-1.
\end{array}\label{eq:cons-1}
\end{eqnarray}
  \textcolor{red}{}

%% file: fgs_nominal_ecc_2.tex
\section{Finite-gain analysis}

In this section, a closed-loop finite-gain stability analysis is developed,
for the case where an exact model of the plant (\ref{eq:plant}) is
available. The case of approximated model is under investigation.
Preliminary results (not reported here) show that the present finite-gain
framework can be extended to the ``robust'' case without substantial
modifications. 

The finite-gain stability concept we consider is the following.

\begin{definition}\label{def:fg_stab}Consider a discrete-time nonlinear
system $x_{t+1}=g(x_{t},w_{t})$, where $x_{t}\in\mathbb{R}^{n_{x}}$
is the state, $x_{0}\in X_{0}\subseteq\mathbb{R}^{n_{x}}$ is the
initial condition, and $w_{t}\in W\subseteq\mathbb{R}^{n_{w}}$ is
the input (containing all external signals acting on the system, including
commands and disturbances). The system is finite-gain stable on $(X_{0},\mathcal{W})$
if a constant $0\leq\rho<1$, and finite non-negative constants $\Lambda_{0},\Lambda_{w},\Lambda_{b}$
exist, such that
\[
\left\Vert x_{t}\right\Vert \leq\Lambda_{0}\rho^{t}\left\Vert x_{0}\right\Vert +\Lambda_{w}\sum_{i=0}^{t-1}\rho^{t-i}\left\Vert w_{i}\right\Vert +\Lambda_{b},\quad t=1,2,\ldots
\]
for all initial conditions $x_{0}\in X_{0}$ and input sequences $(w_{0},w_{1},\dots)\in\mathcal{W}$.
 Stability is said without bias if $\Lambda_{b}=0$. $\square$

\end{definition}

This FGS concept is similar to the standard one, given, e.g., in \cite{book-khalil-2002}. However, it is somewhat more precise,
in the sense that it captures the dependence on initial conditions
and, through the constant $\rho$, the convergence/divergence behavior of the
system. Note that, the ISS concept corresponds to the FGS one when the bias $\Lambda_{b}=0$ and the full state is taken as system's output~\cite{limon2002}.

%\subsection{Finite-gain stability }

A finite-gain stability analysis is now developed for the closed-loop
system given by the plant (\ref{eq:plant}) controlled in feedback
by the NMPC strategy of Section \ref{sec:nmpc}. This system is described
by $x_{t+1}=f(x_{t},u_{t})+d_{k}$, where $u_{t}=S_{1}^{u}\Gamma c^{*}$
and $c^{*}$ is the solution of problem (\ref{eq:opt-1}) at time
$t$. Since this solution depends on $x_{t}$ and the reference sequence
$\mathbf{r}_{t}\doteq(r_{t+1},\ldots,r_{t+T})$, we use the notation
$c^{*}\equiv c^{*}(x_{t},\mathbf{r}_{t})$. The closed-loop system
state equation is
\begin{equation}
\begin{alignedat}{1}x_{t+1} & =f(x_{t},u_{t})+d_{t} =f(x_{t},S_{1}^{u}\Gamma c^{*}(x_{t},\mathbf{r}_{t}))+d_{t}\\
 & =f_{c}(x_{t},\mathbf{r}_{t})+d_{t}
\end{alignedat}
\label{eq:cls1}
\end{equation}
where $f_{c}(x_{t},\mathbf{r}_{t})\doteq f(x_{t},S_{1}^{u}\Gamma c^{*}(x_{t},\mathbf{r}_{t}))$
is the closed-loop transition function. %, $d_{t}\in D\subset\mathbb{R}^{n_{d}}$
%is the disturbance and $D$ is a bounded set. 
The initial state and
each reference sample belong to a set where the penalty functions
are not active: $x_{0}\in X_{0}\doteq\bar{X}_{0}\cap X_{\mathrm{na}}$,
where $\bar{X}_{0}\subset\mathbb{R}^{n_{x}}$ is a compact connected
set and $X_{\mathrm{na}}\doteq\{z\in\mathbb{R}^{n_{x}}:\sigma_{j}(z)=0,\forall j\}$.
Similarly: $r_{t}\in R\doteq\bar{R}\cap X_{\mathrm{na}}$, where $\bar{R}\subset\mathbb{R}^{n_{x}}$
is a compact connected set. The reference sequence $\boldsymbol{r}=(r_{0},r_{1},\dots)$
belongs to a signal set of interest: $\boldsymbol{r}\in\mathcal{R}\subseteq R^{\infty}$,
where $\mathcal{R}\doteq\{(r_{0},r_{1},\dots):r_{k}\in R\}$.

The following lemma shows that 
the dynamics of the tracking error $e_{t}\doteq r_{t}-x_{t}$
can be described by a quasi-LPV (Linear Parameter Varying) equation.
%\footnote{Lemma 1 does not introduce a standard autonomous LTV system. It provides an exact trajectory-dependent affine representation of the tracking-error dynamics where the state matrix is defined point-wise along the closed-loop trajectory through $(x_t, r_t)$.}

\begin{lemma}\label{lem:ltv_err}Consider the closed-loop system
(\ref{eq:cls1}). The dynamics of the tracking error $e_{t}\doteq r_{t}-x_{t}$
is described by
\begin{equation}
\begin{alignedat}{1}e_{t+1} & =A_{t}e_{t}+\xi_{t}-d_{t}\end{alignedat}
\label{eq:treev}
\end{equation}
where $\xi_{t}\doteq r_{t+1}-f_{c}(r_{t},\mathbf{r}_{t})$ and
\begin{equation}
A_{t}\doteq\left\{ \begin{array}{l}
\frac{(f_{c}(r_{t},\mathbf{r}_{t})-f_{c}(x_{t},\mathbf{r}_{t}))(r_{t}-x_{t})^{\top}}{\left\Vert r_{t}-x_{t}\right\Vert ^{2}},\;x_{t}\neq r_{t}\\
\mathbf{0},\;x_{t}=r_{t}.
\end{array}\right.\label{eq:At_mat}
\end{equation}

\end{lemma}

\textbf{Proof.} See the Appendix. $\square$

Besides providing a quasi-LPV representation of the tracking error dynamics,
the lemma shows that this dynamics is driven by two external variables:
the disturbance $d_{t}$ and the the quantity $\xi_{t}$. This latter
can be interpreted as a ``reachability'' variable, measuring the ``smoothness''
of the reference signal. Indeed, suppose for simplicity that the matrix
$R$ in (\ref{eq:obj}) is null. Consider the ideal situation where,
at a time $t$, $x_{t}=r_{t}$ and a $c$ exists such that $r_{t+1}=f(r_{t},S_{1}^{u}\Gamma c^{*}(x_{t},\mathbf{r}_{t}))$.
Then, minimizing the objective function yields $\xi_{t}=0$. In such
an ideal situation, $\xi_{t}$ is null and the tracking error dynamics
is only driven by the disturbance. In generic situations, ``smooth''
reference signals are expected to yield small values of $\xi_{t}$,
while references with abrupt variations give large values of $\xi_{t}$. 

Based on the above lemma, a finite-gain stability result is now derived,
using the following assumption. 

\begin{assumption}\label{assu:tm_bound}

Closed-loop transition matrix bound. Define the transition matrix
\begin{equation}
F_{t}^{\tau}\doteq\left\{ \begin{array}{l}
A_{t+\tau-1}A_{t+\tau-2}\ldots A_{t},\;\tau>0\\
I,\;\tau=0
\end{array}\right.\label{eq:ltv1}
\end{equation}
where $A_{i}$ is given in (\ref{eq:At_mat}). We assume that an integer
$\tau>0$ exists such that 
\begin{equation}
\gamma\equiv\gamma(\tau)\doteq\max\limits_{l\geq0,x_{0}\in X_{0},\boldsymbol{r}\in\mathcal{R}}\left\Vert F_{l}^{\tau}\right\Vert <1.\quad\square\label{eq:stab_cond}
\end{equation}

\end{assumption}

This assumption is based on the quasi-LPV tracking error dynamics (\ref{eq:treev}), whose properties can be studied using the theory of Linear Time Varying (LTV) systems.
Most stability assumptions that can be found in the LTV literature
regard the matrix $A_{l}$. For example, a standard assumption is
that $A_{l}$ has eigenvalues inside the unit circle and is ``slowly
varying'' \cite{Rugh96}. Another typical assumption is $\left\Vert A_{l}\right\Vert <1$,
$l=0,1,\ldots$, \cite{Rugh96}. In our approach, we do not make assumptions
on $A_{l}$ but on the matrix $F_{l}^{\tau}$. This latter is the
product of several subsequent matrices $A_{l}$, making our assumption
weaker than the standard ones regarding directly $A_{l}$. As shown by the following theorem, \eqref{eq:stab_cond} is the key condition for guaranteeing closed-loop finite-gain stability. %Moreover, this condition is the basis of the NMPC design method presented in Section \ref{sec:nmpc_des}.

% \begin{remark}
%     We remark that Assumption \ref{assu:tm_bound} is not itself a finite-gain statement. Rather, it is a finite-step contraction condition on the trajectory-dependent transition matrix $F_t^\tau$. Explicit weighted tracking-error bound for FGS will be derived, later on, in Theorem \ref{thm:fgs_nom}.
% \end{remark}

%The stability result can now be presented. 
\begin{theorem}\label{thm:fgs_nom}Let Assumption \ref{assu:tm_bound}
hold. Then:
\begin{itemize}
\item[(i)] The closed-loop system (\ref{eq:cls1}) is finite-gain stable on
$(X_{0},\mathcal{R})$. 
\item[(ii)] The tracking error $e_{t}\doteq r_{t}-x_{t}$ is bounded as 
\begin{equation}
\begin{alignedat}{1}\left\Vert e_{t}\right\Vert   \leq&\Lambda_{0}\rho^{t}\left\Vert e_{0}\right\Vert +\Lambda_{\xi}\sum_{i=0}^{t-1}\rho^{t-i}\left\Vert \xi_{i}\right\Vert \\
&+\Lambda_{d}\sum_{i=0}^{t-1}\rho^{t-i}\left\Vert d_{i}\right\Vert \end{alignedat}
\label{eq:treb}
\end{equation}
where $\rho\doteq\gamma^{1/\tau}<1$, and $\Lambda_{0}$, $\Lambda_{\xi}$,
$\Lambda_{d}$ are finite non-negative constants.
\end{itemize}
\end{theorem}

\textbf{Proof.} See the Appendix. $\square$

\begin{remark} According to Theorem \ref{thm:fgs_nom}, the key condition for closed-loop finite-gain
stability is $\gamma<1$, see \eqref{eq:stab_cond}. 
As discussed above, this condition relies on a quasi-LPV representation of the tracking error dynamics and is less conservative than the standard assumptions that can be used for quasi-LPV/LTV systems. 
An important observation is that, while checking FGS using Definition \ref{def:fg_stab} is practically impossible, the verification of \eqref{eq:stab_cond} is possible.
Indeed, in Section \ref{sec:nmpc_des} we propose a viable
method for the numerical estimation of $\gamma$, allowing us to verify
stability in practice. The estimate of $\gamma$ is also the core
of a systematic procedure for efficient NMPC design, see Sections~\ref{sec:nmpc_des} and~\ref{sec:results}. $\square$
\end{remark}

% \subsection{Reference smoothing}

% \label{subsec:ref_sm}

% In several situations, it may be convenient to use smooth reference
% signals. Indeed, smooth references are in general expected to give
% smaller values of $\xi_{t}$, $A_{t}$ and $F_{t}^{\tau}$. Moreover,
% there may be situations where a NMPC controller works well even if
% Assumptions A1-A2 are not satisfied (they are sufficient conditions).
% This may happen for example when a constraint suddenly appears and
% abrupt movements must be accomplished (e.g., in obstacle avoidance
% maneuvers) or when the reference is far from the current state and
% a constraint obliges the system to perform a complicated trajectory
% to get close to the reference. 

% In such situations, a smoothed reference $\breve{r}_{t}$ can be constructed
% as follows. Initial value: $\breve{r}_{0}=x_{0}$. For all $t>0$:
% \[
% \begin{alignedat}{1}\breve{r}_{t} & \doteq\arg\min_{\hat{r}\in R}\left\Vert \hat{r}-r_{t}\right\Vert \\
%  & \mathrm{s.t.}:\;\left\Vert \hat{r}-\breve{r}_{t-1}\right\Vert \leq\delta
% \end{alignedat}
% \]
% where $\delta$ is the maximum allowed variation of $\breve{r}_{t}$.
% Then, Theorem \ref{thm:fgs_nom} can be applied with $r_{t}\rightarrow\breve{r}_{t}$,
% and the resulting tracking error is bounded as $\left\Vert e_{t}\right\Vert \leq\bar{e}_{t}+\left\Vert r_{t}-\breve{r}_{t}\right\Vert $,
% where $\bar{e}_{t}$ is the right hand side of (\ref{eq:treb}) computed
% with $r_{t}\rightarrow\breve{r}_{t}$. 

% The effectiveness of this technique is illustrated in \textcolor{red}{the
% numerical example of Section ??}. 

%% file: nmpc_design_2.tex
\section{NMPC design} \label{sec:nmpc_des}
This section presents a systematic NMPC design methodology, allowing an effective selection of a suitable NMPC configuration leading to closed-loop finite-gain stability, as defined by Theorem \ref{thm:fgs_nom}, while providing satisfactory tracking performance. The design procedure consists of three key steps: (i) determination of the integer horizon $\tau$ used in the finite-gain analysis; (ii) a closed-loop function sampling campaign that generates candidate NMPC configurations and closed-loop data over a finite interval of length $\tau$; (iii) a selection rule that chooses the best configuration according to both tracking error and stability criteria.

Let $\Phi$ denote a finite set of NMPC configurations. Each configuration $\phi\in\Phi$ is described by a prediction horizon $T$ and weighting matrices $Q \succ 0$, $R \succ 0$, $P \succ 0$. Let $s=(x_0,\boldsymbol{r},\boldsymbol{d})$ be a generic scenario containing an initial state $x_0\in X_0$, a reference sequence $\boldsymbol{r}=(r_{0},r_{1},\dots)\in\mathcal{R}$ and an exogenous disturbance sequence $\boldsymbol{d}=(d_{0},d_{1},\dots)\in\mathcal{D}$.

At a given time $t$ and for a fixed configuration $\phi$ and scenario $s$, denote: i) $x_t(\phi,s)$  the closed-loop state trajectory produced by the NMPC law; ii) $e_t(\phi,s)$ the corresponding tracking error; and iii) $F_t^{\tau}(\phi,s)$ the transition matrix. 

It is important to highlight that the proposed procedure does not provide a formal guarantee of finite-gain stability for all initial conditions $x_0 \in X_0$. Instead, it offers a practical approach to identify NMPC configurations that are likely to satisfy the finite-gain condition over a predefined operating domain, described through a finite set of representative scenarios. The reliability of the method therefore depends on the coverage of the scenario sets and on the robustness margins introduced in the design.

In the following, the design methodology is described in detail.

\subsection{Selection of the simulation horizon $\tau$}
The integer $\tau$, introduced in Assumption~\ref{assu:tm_bound}, must be chosen sufficiently large to ensure that the norm of the product $F^{\tau}_t(\phi,s)$, denoted by $\gamma(\tau)$, remains strictly less than one for all the considered scenarios. Since the exact verification of this condition over the full domain is generally intractable, we estimate a conservative value of $\tau$ by means of the following procedure:

\begin{enumerate}
  \item Generate an initial scenario set $\mathcal{S}_{\mathrm{in}}=\{s^1,\dots,s^M\}$, where each element is defined as $s^i=(x^i_0,\boldsymbol{r}^i,\boldsymbol{d}^i)$.
  \item Choose a nominal NMPC configuration $\phi_{\mathrm{nom}}$, denoted by the tuple 
  $\phi_{\mathrm{nom}}=(T_{\mathrm{nom}},Q_{\mathrm{nom}},R_{\mathrm{nom}},P_{\mathrm{nom}})$,  
  obtained, for instance, from standard design heuristics such as Bryson's rule \cite{BrysonHo69}.
  \item For each scenario $s^i \in \mathcal{S}_{\mathrm{in}}$ perform a closed-loop simulation, using $\phi_{\mathrm{nom}}$, until a final time $T_{\mathrm{fin}}$ (e.g., the time at which the system reaches the reference or a predefined maximum simulation time). At each time $t=1,\dots,T_{\mathrm{fin}}$ compute the value of $\gamma$ as defined in Eq.  \eqref{eq:stab_cond} and determine the smallest integer $\tau_i$ such that
  \begin{equation} \label{eq:gamma_t}
  %\gamma_i(\tau_i)<1.
  \gamma_i(t)<1,\quad \forall t= \tau_i,...T_{\mathrm{fin}}.
\end{equation}
  \item Compute $\tau^*$ as
  \begin{equation}\label{eq:tau_choice}
 \tau^\star = \max_{i\in\{1,\dots,M\}}\tau_i + \Delta_\tau,
  \end{equation}
    where $\Delta_\tau\ge0$ is a robustness margin. 
\end{enumerate}
The role of $\Delta_\tau$ is to improve the reliability of the finite-gain condition with respect to untested scenarios. Since $\tau^\star$ is computed from a finite scenario set, it may underestimate the true contraction horizon required over the entire set $X_0$. The addition of $\Delta_\tau$ compensates for this effect and increases the likelihood that the finite-gain condition $\gamma(\tau^\star)<1$ also holds for scenarios not explicitly included in $\mathcal{S}_{\mathrm{in}}$. However, this mechanism does not provide a formal guarantee over all $x_0 \in X_0$, and its effectiveness depends on the representativeness of the scenario set.

If~\eqref{eq:tau_choice} cannot be satisfied for any $\tau$, it may be necessary to modify $\phi_{\mathrm{nom}}$ (e.g. increase $T_p$ or retune $Q,R,P$) or relax performance requirements.

\subsection{Closed-loop Function Sampling}
A Closed-loop Function Sampling campaign is conducted to evaluate the finite-time closed-loop response for each NMPC configuration. The steps are:

\begin{enumerate}
  \item Generate an evaluation scenario set $\mathcal{S}=\{s^1,\dots,s^N\}$, with $N \ge M$. This set is typically constructed in order to ensure a proper coverage of the relevant initial states, reference trajectories, and disturbance profiles, and may be an an extension of the initial scenario set, i.e., $\mathcal{S}\supseteq\mathcal{S}_{\mathrm{in}}$.
  \item Build the NMPC configuration set
  $
    \Phi=\{\phi_j\}_{j=1}^K,\, \phi_j=(T^{(j)},Q^{(j)},R^{(j)},P^{(j}),
  $
  by choosing different values of $T$, $Q,R$ and $P$. The matrices $Q$, $R$ and $P$ can be scaled starting from the nominal NMPC configuration $ Q=\alpha_Q Q_{\mathrm{nom}},\,
  R=\alpha_R R_{\mathrm{nom}},\,
  P=\alpha_P P_{\mathrm{nom}}$   
where $\alpha_Q,\alpha_R,\alpha_P>0$ are scalar tuning parameters.
  \item For each configuration $\phi_j$, simulate the closed-loop state trajectory $x_t(\phi_j,s^i)$ over the time interval $t=1,\dots,\tau^*$ for all scenarios $s^i$, and store the resulting tracking error $e_{\tau^*}(\phi_j,s^i)$ and transition matrix $F^{\tau^*}_t(\phi_j,s^i)$.
\end{enumerate}

\subsection{Selection of the optimal NMPC configuration}
For each NMPC configuation $\phi \in \Phi$, we define two metrics:

\paragraph{Tracking error metric}
The worst-case tracking error over all scenarios is defined as
\begin{equation}\label{eq:E_c}
E(\phi)=\max_{i=1,\dots,N} E(\phi,s^i)=\max_{i=1,\dots,N}\|e_{\tau^*}(\phi,s^i)\|.
\end{equation}

\paragraph{Finite-gain stability metric}
The worst-case finite-gain stability index is given by
\begin{equation}\label{eq:L_c}
L(\phi)=\max_{i=1,\dots,N} \gamma_i(\tau^*)=\max_{i=1,\dots,N}\|\big(F^{\tau^*}_t(\phi,s^i)\big)\|.
\end{equation}

For the selection rule, we first normalize the two metrics relative to their worst-case values across the entire set $\Phi$
\[
    E_n(\phi)=\frac{E(\phi)}{\max_{\Phi}E},\qquad
    L_n(\phi)=\frac{L(\phi)}{\max_{\Phi}L},
  \]
and then define the following cost function 
  \begin{equation}
    J_{sr}(\phi)=\alpha_J E_n(\phi)+(1-\alpha_J) L_n(\phi),\qquad \alpha_J\in[0,1].
  \end{equation}
  Finally, the optimal NMPC configuration is selected as
  \begin{equation}\label{eq:Jc}
  \phi^\star=\arg\min_{\phi\in\Phi}J_{sr}(\phi).
  \end{equation}

\subsection{Visualization and sensitivity analysis}
To improve interpretability of the results, the sensitivity of the finite-gain stability index (and tracking performance) to key design parameters can be analyzed throug parametric analysis.
Specifically, surface or level curves plots of $L(\phi)$ and $E(\phi)$ are generated over a two-dimensional grid defined by the prediction horizon $T$ and the ratio $\kappa = \alpha_Q/\alpha_R$.
%$\kappa = \mathrm{Tr}(Q)/\mathrm{Tr}(R)$. 
These plots show how variations in $T$ and in the relative weighting between states and inputs influence the closed-loop stability margin (and tracking accuracy).

\subsection{Summary of the design algorithm}
Algorithm~\ref{alg:design} summarizes the proposed procedure.
\begin{algorithm}[htbp]
\caption{Selection of the optimal NMPC configuration via finite-gain analysis}
\label{alg:design}

\textbf{Input:} System \eqref{eq:plant}, $X_0$, $\mathcal{R}$, $\mathcal{D}$

\textbf{Output:} $\phi^\star$

\begin{algorithmic}[1]

  \State Choose nominal $\phi_{\mathrm{nom}}$ and compute minimal $\tau^*$ satisfying \eqref{eq:tau_choice}.
  \State Construct the NMPC configuration set $\Phi$ and the scenario set $\mathcal{S}$.
  \ForAll{$\phi\in \Phi$}:
  \ForAll{$s^i\in\mathcal{S}$}:
      \State Simulate closed-loop for $t=1,\dots,\tau^*$.
      \State Compute $E(\phi,s^i)$ and $F^{\tau^*}_t(\phi,s^i)$.
  \EndFor
    \State Compute $E(\phi)$ and $L(\phi)$.
    \EndFor
  \State Select $\phi^\star$ through Eq.~\eqref{eq:Jc}.
  \State Produce surface plots and, if necessary, refine the candidate grid.
\end{algorithmic}
\end{algorithm}

By exploiting the finite-gain stability condition of Theorem~\ref{thm:fgs_nom}, the proposed procedure provides a systematic and computationally efficient method for obtaining an optimal NMPC configuration through a function sampling campaign in which each closed-loop simulation is performed only over the horizon $\tau^*$. This significantly reduces the computational effort compared to traditional long-horizon simulations, where different types of scenarios are considered and, for each of them, several maneuvers are simulated for their complete duration.